\documentclass{birkjour}

\usepackage{amssymb,amsthm,amsmath,amsfonts}
\usepackage{tikz}
\usepackage{subcaption}
\usepackage{array}
\usepackage{pifont}
\usepackage{pdfpages}

\newtheorem{thm}{Theorem}[section]
\newtheorem{prop}[thm]{Proposition}
\newtheorem{cor}[thm]{Corollary}
\newtheorem{lem}[thm]{Lemma}

\newtheorem{defn}[thm]{Definition}

\DeclareMathOperator{\ddp}{dp}
\DeclareMathOperator{\DP}{DP}
\DeclareMathOperator{\N}{\mathcal{N}}

\newcolumntype{P}[1]{>{\centering\arraybackslash}p{#1}}

\begin{document}

%\includepdf{Response_to_reviewer.pdf}
%\setcounter{page}{1}

\title{On Distance Preserving and Sequentially \linebreak Distance Preserving Graphs}
\author{Jason P. Smith\textsuperscript{*}}
\address{Department of Physics and Mathematics,\\ Nottingham Trent University,\\ Nottingham, UK.}
\email{jason.smith@ntu.ac.uk}
\thanks{J.P. Smith was supported by the EPSRC Grant EP/M027147/1}
\thanks{\textsuperscript{*} Correspoding Author}

\author{Emad Zahedi}
\address{Department of Mathematics, Michigan State University, \\
              East Lansing, MI 48824-1027, U.S.A. \\
              Department of Computer Science and Engineering,	Michigan State University,\\
							East Lansing, MI 48824, U.S.A.}
\email{Zahediem@msu.edu}

\begin{abstract}
A graph $H$ is an \emph{isometric} subgraph of $G$ if $d_H(u,v)= d_G(u,v)$, for every pair~$u,v\in V(H)$. A graph is \emph{distance preserving} if it has an isometric subgraph of every possible order. A graph is \emph{sequentially distance preserving} if its vertices can be ordered such that deleting the first $i$ vertices results in an isometric subgraph, for all $i\ge1$. We give an equivalent condition to sequentially distance preserving based upon simplicial orderings. Using this condition, we prove that if a graph does not contain any induced cycles of length~$5$ or greater, then it is sequentially distance preserving and thus distance preserving. Next we consider the distance preserving property on graphs with a cut vertex. Finally, we define a family of non-distance preserving graphs constructed from cycles. 
\end{abstract}


\maketitle

\section{Introduction}\label{Introduction} 

The distance between two vertices in a graph plays an important role in many areas of graph theory. 
 Moreover, computing graph distances is integral to many real-world applications, especially in areas such as optimisation theory. 
 Computing distances between vertices in large graphs is extremely expensive, such as social networks with millions of vertices. 
 It is often desirable to know the distances between vertices in subgraphs of the original graph, yet this requires recomputing all the distances. 
 One solution to this problem is to consider subgraphs where the distances between all vertices is equal to their distance in the original graph, 
 such a subgraph is called \emph{isometric}.

In this framework all graphs are finite, non-empty, simple and connected, unless assumed otherwise. 
 A graph  $G$ is \emph{distance preserving}  if $G$ has an isometric subgraph of every possible order. 
 The notion of distance preserving graphs is a generalisation of \emph{distance-hereditary} graphs, 
 where a  graph is distance-hereditary if every connected induced subgraph is isometric. 
 Distance-hereditary graphs where first introduced by Howorka in~\cite{howorka1977characterization} and have since been studied in various papers, 
 see~\cite{bandelt1986distance, damiand2001simple, hammer1990completely}. 
 Distance-hereditary graphs have many nice properties, for example they are known to be perfect graphs \cite{golumbic2000clique}, and many NP-hard problems have
 polynomial time solutions on distance-hereditary graphs, such as finding dominating sets \cite{d1988distance}, 
 Hamiltonian cycles \cite{hsieh2002efficient}, and optimal communication trees \cite{esfahanian1993distance}.
 The set of distance-hereditary graphs is a subset of distance preserving graphs, and the less restrictive definition of distance preserving graphs allows for more complex structure.
 
Isometric subgraphs play an important role in metric graph theory, 
 which is the study of classes of graphs satisfying certain properties of classical metric geometries, see \cite{bandelt2008metric}.
 Metric graph theory, and thus isometric subgraphs, has important applications in a variety of areas, including geometric group theory,
 concurrency and learning theory, and combinatorial optimisation \cite{linial1995geometry}.

Distance preserving graphs where introduced in \cite{nussbaum2013clustering}, along with a clustering algorithm which partitions a graph into distance preserving subgraphs.
 The algorithm is applied to real-world social networks, where it is shown that clustering based on the distance preserving property is an effective way to extract communities from 
 large networks. This was further supported in \cite{Nus14} where it was shown that known communities within social networks form subgraphs 
 that are distance preserving, or almost distance preserving, 
 such as UK members of parliament within the Twitter network.
 However, limitations are encountered with this clustering approach, in part due to the lack of a theoretical understanding of distance preserving graphs.

Some of the first theoretical results appeared in \cite{NF2}, along with a variety of conjectures, including that almost all graphs are distance preserving.
 The key result of \cite{NF2} is that if an $n$-vertex graph $G$ has minimum degree  $\delta(G)\ge\frac{2n}{3}-1$, then $G$ is distance preserving,
 and it is conjectured that this bound can be lowered to $\delta(G)\ge\frac{n}{2}$, which was partially proved in \cite{kha20} for all graphs over a certain size threshold.
 In \cite{khalifeh2015distance,Zah19} results are given on the behaviour of the distance preserving properties under the Cartesian and lexicographic graph products and modular decomposition.
 A construction is given for distance preserving regular graphs of all possible orders and degrees of regularity in \cite{esfahanian2014constructing}.
 Algorithmic results were developed in \cite{bernstein2019distance,zahedi2017evolutionary} for finding isometric and almost isometric subgraphs.

One way to show that a graph $G$ is distance preserving is to give an ordering of the vertices~$v_1,\ldots,v_n$ of $G$, 
 such that removing $v_1,\ldots,v_i$ results in an isometric subgraph, for all $i\ge 1$. If such an ordering exists we say that~$G$ is 
 \emph{sequentially distance preserving}. Note that every distance-hereditary graph is sequentially distance preserving, 
 and every sequentially distance preserving graph is distance preserving. See Figure~\ref{fig:examples} for some examples of graphs satisfying these properties. 

\begin{figure}
\begin{subfigure}{0.49\textwidth}
\begin{center}
    \begin{tikzpicture}	[thick, scale=0.7]
    
        %\node[label=left:{$7$}] (1) at (0,-0.6){};
	    \node (0) at (0,0){};
		\node (a1) at (1,1){};
    	\node (a2) at (1,1.5){};
        \node (a3) at (1.5,1){};
		\node (b1) at (-1,1){};
    	\node (b2) at (-1,1.5){};
        \node (b3) at (-1.5,1){};
		\node (c1) at (1,-1){};
    	\node (c2) at (1,-1.5){};
        \node (c3) at (1.5,-1){};
		\node (d1) at (-1,-1){};
    	\node (d2) at (-1,-1.5){};
        \node (d3) at (-1.5,-1){};
        
        \draw[line width=.3pt] (0) to (a1);
        \draw[line width=.3pt] (0) to (a2);
        \draw[line width=.3pt] (0) to (a3);
        \draw[line width=.3pt] (0) to (b1);
        \draw[line width=.3pt] (0) to (b2);
        \draw[line width=.3pt] (0) to (b3);
        \draw[line width=.3pt] (0) to (c1);
        \draw[line width=.3pt] (0) to (c2);
        \draw[line width=.3pt] (0) to (c3);
        \draw[line width=.3pt] (0) to (d1);
        \draw[line width=.3pt] (0) to (d2);
        \draw[line width=.3pt] (0) to (d3);
        \draw[line width=.3pt] (a1) to (a2);
        \draw[line width=.3pt] (a2) to (a3);
        \draw[line width=.3pt] (a3) to (a1); 
        \draw[line width=.3pt] (b1) to (b2);
        \draw[line width=.3pt] (b2) to (b3);
        \draw[line width=.3pt] (b3) to (b1); 
        \draw[line width=.3pt] (c1) to (c2);
        \draw[line width=.3pt] (c2) to (c3);
        \draw[line width=.3pt] (c3) to (c1); 
        \draw[line width=.3pt] (d1) to (d2);
        \draw[line width=.3pt] (d2) to (d3);
        \draw[line width=.3pt] (d3) to (d1);    
                     
    	\shade[shading=ball, ball color=black] (0) circle (.1/0.7);
    	\shade[shading=ball, ball color=black] (a1) circle (.1/0.7);
	    \shade[shading=ball, ball color=black] (a2) circle (.1/0.7);
	    \shade[shading=ball, ball color=black] (a3) circle (.1/0.7);
		\shade[shading=ball, ball color=black] (b1) circle (.1/0.7);
    	\shade[shading=ball, ball color=black] (b2) circle (.1/0.7);
		\shade[shading=ball, ball color=black] (b3) circle (.1/0.7);
    	\shade[shading=ball, ball color=black] (c1) circle (.1/0.7);
	    \shade[shading=ball, ball color=black] (c2) circle (.1/0.7);
	    \shade[shading=ball, ball color=black] (c3) circle (.1/0.7);
		\shade[shading=ball, ball color=black] (d1) circle (.1/0.7);
    	\shade[shading=ball, ball color=black] (d2) circle (.1/0.7);
		\shade[shading=ball, ball color=black] (d3) circle (.1/0.7);
	\end{tikzpicture}
\end{center}\caption{Windmill Graph $W(4,4)$}
\end{subfigure}
\begin{subfigure}{0.49\textwidth}
\begin{center}
    \begin{tikzpicture}	[thick, scale=0.7]
    
        %\node[label=left:{$7$}] (1) at (0,-0.6){};
	    \node (11) at (1,1){};
	    \node (12) at (1,2){};
		\node (13) at (1,3){};
    	\node (14) at (1,4){};
	    \node (21) at (2,1){};
	    \node (22) at (2,2){};
		\node (23) at (2,3){};
    	\node (24) at (2,4){};
	    \node (31) at (3,1){};
	    \node (32) at (3,2){};
		\node (33) at (3,3){};
    	\node (34) at (3,4){};
	    \node (41) at (4,1){};
	    \node (42) at (4,2){};
		\node (43) at (4,3){};
    	\node (44) at (4,4){};

        \draw[line width=.3pt] (11) to (12);
        \draw[line width=.3pt] (12) to (13);
        \draw[line width=.3pt] (13) to (14);
        \draw[line width=.3pt] (21) to (22);
        \draw[line width=.3pt] (22) to (23);
        \draw[line width=.3pt] (23) to (24);
        \draw[line width=.3pt] (31) to (32);
        \draw[line width=.3pt] (32) to (33);
        \draw[line width=.3pt] (33) to (34);
        \draw[line width=.3pt] (41) to (42);
        \draw[line width=.3pt] (42) to (43);
        \draw[line width=.3pt] (43) to (44); 
        \draw[line width=.3pt] (11) to (21);
        \draw[line width=.3pt] (21) to (31);
        \draw[line width=.3pt] (31) to (41);                            
        \draw[line width=.3pt] (12) to (22);
        \draw[line width=.3pt] (22) to (32);
        \draw[line width=.3pt] (32) to (42);               
        \draw[line width=.3pt] (13) to (23);
        \draw[line width=.3pt] (23) to (33);
        \draw[line width=.3pt] (33) to (43);
        \draw[line width=.3pt] (14) to (24);
        \draw[line width=.3pt] (24) to (34);
        \draw[line width=.3pt] (34) to (44);                 

    	\shade[shading=ball, ball color=black] (11) circle (.1/0.7);
	    \shade[shading=ball, ball color=black] (12) circle (.1/0.7);
	    \shade[shading=ball, ball color=black] (13) circle (.1/0.7);
		\shade[shading=ball, ball color=black] (14) circle (.1/0.7);
    	\shade[shading=ball, ball color=black] (21) circle (.1/0.7);
		\shade[shading=ball, ball color=black] (22) circle (.1/0.7);
    	\shade[shading=ball, ball color=black] (23) circle (.1/0.7);
    	\shade[shading=ball, ball color=black] (24) circle (.1/0.7);
    	\shade[shading=ball, ball color=black] (31) circle (.1/0.7);
	    \shade[shading=ball, ball color=black] (32) circle (.1/0.7);
	    \shade[shading=ball, ball color=black] (33) circle (.1/0.7);
		\shade[shading=ball, ball color=black] (34) circle (.1/0.7);
    	\shade[shading=ball, ball color=black] (41) circle (.1/0.7);
		\shade[shading=ball, ball color=black] (42) circle (.1/0.7);
    	\shade[shading=ball, ball color=black] (43) circle (.1/0.7);
    	\shade[shading=ball, ball color=black] (44) circle (.1/0.7);
	\end{tikzpicture}
\end{center}\caption{Grid graph $G(4,4)$}
\end{subfigure}
\vskip 20pt
\begin{subfigure}{0.49\textwidth}
\begin{center}
    \begin{tikzpicture}	[thick, scale=0.7]
    
        %\node[label=left:{$7$}] (1) at (0,-0.6){};
	    \node (1) at (0:1cm){};
	    \node (2) at (51:1cm){};
		\node (3) at (102:1cm){};
    	\node (4) at (153:1cm){};
        \node (5) at (204:1cm){};
        \node (6) at (256:1cm){};
        \node (7) at (307:1cm){};

        \draw[line width=.3pt] (1) to (2) ;
        \draw[line width=.3pt] (2) to  (3);
        \draw[line width=.3pt] (3) to  (4);
        \draw[line width=.3pt] (4) to  (5);
        \draw[line width=.3pt] (5) to  (6);
        \draw[line width=.3pt] (6) to  (7);
        \draw[line width=.3pt] (7) to  (1);
        \draw[line width=.3pt] (4) to  (1);

    	\shade[shading=ball, ball color=black] (1) circle (.1/0.7);
	    \shade[shading=ball, ball color=black] (2) circle (.1/0.7);
	    \shade[shading=ball, ball color=black] (3) circle (.1/0.7);
		\shade[shading=ball, ball color=black] (4) circle (.1/0.7);
    	\shade[shading=ball, ball color=black] (5) circle (.1/0.7);
		\shade[shading=ball, ball color=black] (6) circle (.1/0.7);
    	\shade[shading=ball, ball color=black] (7) circle (.1/0.7);
	\end{tikzpicture}
\end{center}\caption{$7$-Cycle with a chord $C_7\cup e$}
\end{subfigure}
\begin{subfigure}{0.49\textwidth}
\begin{center}
    \begin{tikzpicture}	[thick, scale=0.7]
    
        %\node[label=left:{$7$}] (1) at (0,-0.6){};
	    \node (1) at (0:1cm){};
	    \node (2) at (70:1cm){};
		\node (3) at (140:1cm){};
    	\node (4) at (210:1cm){};
        \node (5) at (280:1cm){};

        \draw[line width=.3pt] (1) to (2) ;
        \draw[line width=.3pt] (2) to  (3);
        \draw[line width=.3pt] (3) to  (4);
        \draw[line width=.3pt] (4) to  (5);
        \draw[line width=.3pt] (5) to  (1);

    	\shade[shading=ball, ball color=black] (1) circle (.1/0.7);
	    \shade[shading=ball, ball color=black] (2) circle (.1/0.7);
	    \shade[shading=ball, ball color=black] (3) circle (.1/0.7);
		\shade[shading=ball, ball color=black] (4) circle (.1/0.7);
    	\shade[shading=ball, ball color=black] (5) circle (.1/0.7);
	\end{tikzpicture}
\end{center}\caption{Cycle Graph $C_5$}
\end{subfigure}
\vskip 20pt
\begin{subfigure}{\textwidth}
\begin{tabular}{c|cP{9em}c}
 & {\small Distance-Hereditary} & {\small Sequentially Distance Preserving} & {\small Distance Preserving}\\\hline
 $W(4,4)$    & \ding{51}  & \ding{51} & \ding{51} \\
 $G(4,4)$    & \ding{55}  & \ding{51} & \ding{51} \\
 $C_7\cup e$ & \ding{55}  & \ding{55} & \ding{51} \\
 $C_5$       & \ding{55}  & \ding{55} & \ding{55}   \\
\end{tabular}\caption{Distance properties satisfied by the above graphs}
\end{subfigure}
\caption{}\label{fig:examples}
\end{figure}

Few results exist for sequentially distance preserving graphs. The concept was first introduced in \cite{chepoi1998distance} under the name \emph{distance-preserving orderings}, 
 and considered in relation to domination elimination orderings, which were introduced in \cite{dahihaus1995domination}. 
 In \cite{chepoi1998distance} it is shown that pseudo-modular and house-free weakly modular graphs are sequentially distance preserving. 
 The lexicographic product of two sequentially distance preserving graphs is shown to be sequentially distance preserving in \cite{khalifeh2015distance}. 
 It was shown in \cite{coudert18} that determining if a sequentially distance preserving ordering exists is NP-complete.

A graph is \emph{$k$-chordal} if the largest induced cycle is of length $k$. It was shown in~\cite{Zah15} that $3$-chordal graphs are sequentially distance preserving. 
 This is proved using the fact that all $3$-chordal graphs have a certain type of ordering of the vertices called a simplicial ordering.  
 A generalisation of simplicial orderings is introduced in~\cite{Kri13}. 
 We apply this generalisation in Section~\ref{sec:4chordal} to show that $4$-chordal graphs are sequentially distance preserving.

A connected graph has a \emph{cut vertex} $x$ if removing $x$ disconnects the graph.  
 In Section~\ref{Separable graphs} we consider graphs of the form $G\cup H$, where $G$ and $H$ are non-trivial graphs 
 with $E(G)\not=\emptyset$ and $E(H)\not=\emptyset$ and have exactly one common vertex $x$, so $x$ is a cut vertex. 
 We denote such graphs by $G+_{x}H$. We characterise the distance preserving property in $G+_{x} H$ in terms of $G$ and $H$, 
 which reduces the complexity of testing if such graphs are distance preserving.
 Finally, in Section~\ref{$C_{k,l}$} we study the class of non-distance preserving graphs, presenting a result on how to add vertices to cycle graphs 
 whilst maintaining the non-distance preserving property.

%KNOWN RESULTS:
%\begin{enumerate}
%\item \cite{kha20}: partially prove a conjecture of \cite{NF2} that every graph with minimum degree $\delta(G)>\frac{n}{2}$ is dp, they show it is true if for all graphs over a size threshold, still open in general.
%\item \cite{bernstein2019distance} an algorithm is given for edge contractions to preserve pair-wise distances, within some tolerance
%\item Called distance-preserving elimination orderings in \cite{coudert18} where it is shown that it is NP-complete to determine if such an ordering exists,
%\item \cite{zahedi2017evolutionary} an algorithm is given for finding isometric subgraphs
%\item \cite{Zah19} give conditions for a graph to be dp based on a modular decomposition of the graph.
%\item A good overview of dp and sdp are the theses \cite{Nus14,Zahedi2017distance}
%\item \cite{khalifeh2015distance} gave a necessary and sufficient condition for the lexicographic product of two graphs to be dp
%\item \cite{esfahanian2014constructing} constructed regular distance-preserving graphs of all possible orders and degrees of regularity
%\item \cite{NF2} first look with a few preliminary results and many conjectures 
%\item \cite{nussbaum2013clustering} introduce a clustering algorithm based on clustering into dp subgraphs, and demonstrate its application to social networks.
%\end{enumerate}

%%%%%%%%%%%%%%%%%%%%%%%%%%%%%%%%%%%%%%%%%%%%%%%%%%%%%%%%%%%%%%%%%%%%%%%%%%%%%%%%%%%%%%%%%%

\section{Background}\label{Background}

In this section we recall some necessary graph theory concepts. For any definitions and notation not given here, and a general overview of graph theory, we refer the reader to \cite{bondy1976graph}. Let $G$ be a graph with vertex set $V(G)$ and edge set~$E(G)$. For ease of notation, we let $|G|$ be the number of vertices of $G$. A {\em path} in $G$ is a sequence of distinct vertices $v_0, \dots, v_k$ such that $v_iv_{i+1}\in E(G)$, for all $i=0,\dots, k-1$.  The length of a path is $k$, the number of edges. A path~$P$ is \emph{chordless} if there is no edge of $G$ between any non-consecutive pair of vertices of $P$. The \emph{interior} of a path $P$ is obtained by removing the end points $v_0,v_k$ from $P$. The {\em distance} between two vertices $u,v$ in $G$, denoted~$d_G(u,v)$, is the minimum length of a path between these vertices. If $G$ is clear from  context, we will use~$d(u,v)$, instead of~$d_G(u,v)$. A path from $u$ to $v$ with length $d_G(u,v)$ is called a {\em $u$--$v$  geodesic} path.

An induced subgraph $H$ of $G$ is called an {\em isometric} subgraph, denoted~$H\le G$, if $d_H(a,b)=d_G(a,b)$, for every pair of vertices $a,b \in V(H)$. We say that $G$ is {\em distance preserving}, for which we use the abbreviation dp, if there is an $i$-vertex isometric subgraph, for every $1 \leq i \leq |G|$. 
Given a set $A\subseteq V(G)$, let $G[A]$ be the graph induced on the set $A$ and $G- A:=G[V(G)\setminus A]$. We say that $G$ is {\em sequentially distance preserving}, which we abbreviate to sdp, if there is an ordering $v_1, \ldots , v_n$ of~$V(G)$ such that deleting the first $i$ vertices results in an isometric subgraph for all~$i\ge 1$. 

 The cycle graph $C_k$ is the graph with vertices $v_1,\ldots,v_k$ and the edge set $E(C_k):=~\{v_iv_j\,:\,|i-j|=~1\}\cup\{v_1v_k\}$. If $G$ contains $C_k$ as a subgraph we say it contains a cycle of length $k$ or a \emph{$k$-cycle}. If $G$ is connected, then a vertex $v\in V(G)$ is called a \emph{cut vertex} if $G-\{v\}$ is not connected.  The graph $G^\ell$ is the graph whose vertices are those of $G$ and there is an edge between any two vertices~$u,v \in G$ with $d_G(u,v) \le \ell$.  The set of vertices adjacent to $v\in V(G)$ is called its {\em open neighbourhood} and is denoted $\N_G(v)$.  The {\em closed neighbourhood} of $v$ is~$\N_G[v] = \N_G(v) \cup \{v\}$. A \emph{clique} in $G$ is an induced subgraph that has an edge between every pair of its vertices.

%%%%%%%%%%%%%%%%%%%%%%%%%%%%%%%%%%%%%%%%%%%%%%%%%%%%%%%%%%%%%%%%%%%%%%%%%%%%%%%%%%%%%%%%%%

\section{All $4$-chordal graphs are distance preserving}\label{sec:4chordal}

It was shown in \cite{Zah15} that $3$-chordal graphs, often just called chordal graphs, are sequentially distance preserving. This is shown using the well known property that all chordal graphs have a simplicial ordering. This property is generalised to $k$-chordal graphs in~\cite{Kri13}, using the notion of a~$k$-simplicial ordering.

\begin{defn}
A vertex $v$ of a graph $G$ is \emph{weakly $k$-simplicial} if $\N_G(v)$ induces a clique in $(G-\{v\})^{k-2}$. Furthermore, $v$ is \emph{$k$-simplicial} if it is weakly $k$-simplicial and for each non-adjacent pair $x,y$ in $\N_G(v)$, every chordless~$x,y$-path whose interior is entirely in~$G- \N_G[v]$ has at most $k-2$ edges. A vertex ordering $v_1,\ldots,v_{n}$ of $G$ is a \emph{(weakly) $k$-simplicial ordering} if $v_i$ is (weakly) $k$-simplicial in~$G[\{v_i,\ldots,v_{n}\}]$.
\end{defn}

We use this generalised simplicial ordering to prove the conjecture in~\cite{NF2} that all $4$-chordal graphs are distance preserving. In order to do this we need the main result from~\cite{Kri13}, which we present next. Note that there is a third equivalent statement in the original theorem which we omit here as we do not require it for our results.

\begin{thm}\label{thm:kri}~\cite[Theorem 1]{Kri13}
Consider a graph $G$ and integer $k\ge 3$. The graph $G$ is~$k$-chordal if and only if $G$ has a $k$-simplicial ordering.
\end{thm}

Before proving the main result of this section we present the following lemma, which is a generalisation of Lemma $3.1$ of~\cite{Zah15}. 

\begin{lem}\label{lem:w4simp}
Consider a graph $G$ and vertex $v\in V(G)$. The graph $G-{v}$ is isometric if and only if $v$ is weakly $4$-simplicial.
\begin{proof}
Suppose $v$ is weakly $4$-simplicial. This implies that $\N_G(v)$ induces a clique in $(G-{v})^2$, that is, any pair $x,y\in \N_G(v)$ have a distance of at most~$2$ in $G-{v}$. Consider any path $P$ which contains $v$ in its interior. There must be a subpath $x-v-y$ of $P$, where $x,y\in \N_G(v)$. Because $v$ is weakly~$4$-simplicial we know that $x$ and $y$ are either neighbours or have a common neighbour~$z\not=v$. Therefore, we can either remove $v$ or replace it with $z$ to get a path that is at least as short as $P$ lying in $G-v$. It follows that $G-{v}$ is isometric.

Suppose $G-{v}$ is isometric. Consider any pair $u,w\in \N_G(v)$, then we know that $d_G(u,w)\le 2$ which implies $d_{G-{v}}(u,w)\le 2$. Therefore, $\N_G(v)$ induces a clique in $(G-{v})^2$, so $v$ is weakly $4$-simplicial. 
\end{proof}
\end{lem}

The following proposition is an immediate result of Lemma~\ref{lem:w4simp}. 

\begin{prop}\label{prop:sdp}
A graph is sdp if and only if it admits a weakly $4$-simplicial ordering.
\begin{proof}
Lemma \ref{lem:w4simp} implies that a vertex ordering is a weakly $4$-simplicial ordering if and only if it is an sdp ordering.
\end{proof}
\end{prop}

Now we have all we need to prove Conjecture $5.2$ of~\cite{NF2}:
\begin{thm}\label{thm:sdp}
Any $4$-chordal graph is sdp, and thus dp.
\begin{proof}
Applying Theorem~\ref{thm:kri} with $k=4$  shows that for any~$4$-chordal graph there is a $4$-simplicial ordering of the vertices. Moreover, Proposition~\ref{prop:sdp} implies this ordering is an sdp ordering.
\end{proof}
\end{thm}

\begin{figure}\begin{center}
    \begin{tikzpicture}	[thick, scale=0.7]

        \draw[line width=.3pt] (0,-0.6) to (0,.8) ;
        \draw[line width=.3pt] (0,-0.6) to  (4.5,.3);
        \draw[line width=.3pt] (0,-0.6) to  (-2,-1.5);
        \draw[line width=.3pt] (0,.8) to  (-2,-1.5);
        \draw[line width=.3pt] (-2,-1.5) to  (-2,1.5);
        \draw[line width=.3pt] (-2,-1.5) to  (2,-2.5);
        \draw[line width=.3pt] (-2,1.5) to  (2,2.5);
        \draw[line width=.3pt] (2, 2.5) to  (0,.8);
        \draw[line width=.3pt] (2, 2.5) to  (4.5,.3);
        \draw[line width=.3pt] (2, -2.5) to  (4.5,.3);

	%outer
    	\shade[shading=ball, ball color=black] (0,-0.6) circle (.1/0.7);
	    \shade[shading=ball, ball color=black] (0,.8) circle (.1/0.7);
	    \shade[shading=ball, ball color=black] (2,2.5) circle (.1/0.7);
	
		\shade[shading=ball, ball color=black] (2,-2.5) circle (.1/0.7);
    	\shade[shading=ball, ball color=black] (-2,1.5) circle (.1/0.7);
	
	    \shade[shading=ball, ball color=black] (-2,-1.5) circle (.1/0.7);
	    \shade[shading=ball, ball color=black] (4.5,.3) circle (.1/0.7);
	    
	%labels
		\node[label=left:{$7$}] (1) at (0,-0.6){};
	    \node[label=left:{$6$}] (2) at (0,.8){};
	    \node[label=right:{$3$}] (3) at (2,2.5){};
	
		\node[label=right:{$1$}] (4) at (2,-2.5){};
    	\node[label=left:{$4$}] (5) at (-2,1.5){};
	
	    \node[label=left:{$5$}] (6) at (-2,-1.5){};
	    \node[label=right:{$2$}] (7) at (4.5,.3){};
	
	\end{tikzpicture}
\end{center}
\caption{A non-$4$-chordal graph that is sdp. The vertex labels give an sdp ordering.}
\label{fig:C5sdp}
\end{figure}
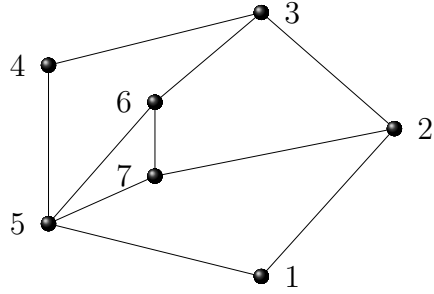

The graph in Figure~\ref{fig:C5sdp} is not $4$-chordal, because it contains an induced $5$-cycle, so by Theorem \ref{thm:kri} the graph cannot have a $4$-simplicial ordering. However, the ordering given by the vertex labels is a weakly $4$-simplicial ordering, so the graph is sdp. To see the ordering is not~$4$-simplicial, note that the vertex labelled $1$ is not~$4$-simplicial because the path $2-3-4-5$ violates the $4$-simplicial condition. Theorem~\ref{thm:sdp} implies that a graph that is dp but not sdp cannot contain an induced $4$-cycle, combining this with \cite[Corollary~3.2]{Zah15} gives the following corollary:

\begin{cor}
Any dp graph that is not sdp must contain an induced cycle of length $k\ge 5$.
\end{cor}

%%%%%%%%%%%%%%%%%%%%%%%%%%%%%%%%%%%%%%%%%%%%%%%%%%%%%%%%%%%%%%%%%%%%%%%%%%%%%%%%%%%%%%%%%%%%

\section{Separable graphs}\label{Separable graphs}
A connected graph is said to  be {\em separable} if it can be disconnected by removing a vertex, which we call a \emph{cut vertex}.
In this section we consider the distance preserving property in separable graphs. A separable graph can be represented in the following way:

\begin{defn}
Consider two non-trivial graphs $G$ and $H$, with $E(G)\not=\emptyset$ and~$E(H)\not=\emptyset$, with a single common vertex~$x$. Let $G+_x H$ be the union of~$G$ and $H$.
\end{defn}

%The graph $G+_x H$ is a separable graph with cut vertex $x$.
So $G +_x H$ is a separable graph with a cut vertex $x$. We characterise the isometric subgraphs of $G +_x H$. To do this we introduce the following lemma.

\begin{lem}\label{H'+G'} 
Consider a graph $G +_x H$ and two induced subgraphs $ H' \subseteq H $, $ G' \subseteq G $, with~$x\in V(G')\cap V(H')$, then:
$$
G'+_x H' \le G +_x H\ \text{   if and only if } \ H'\le H \ \text{and}\ G'\le G.
$$
\end{lem}

\begin{proof}
First we consider the forward direction. Since $x$ is a cut vertex any geodesic path between a pair of vertices of $H\subseteq G +_x H$ is contained in~$H$, thus~$H\leq  G +_x H$. The same is true when replacing $G$ and $H$ by $G'$ and $H'$, respectively. Combining this with our assumption we have:
$$ d_{H'}(u,v) = d_{G' +_x H'} (u,v)=d_{G +_x H}(u,v) = d_H(u,v),$$
for every pair of vertices $u,v\in V(H')$, so $H'\le H$. An analogous argument shows that~$G'\le~G$.

Now consider the backward direction.
Using the fact $H\leq  G +_x H$ and the assumption $H'\le H$, we have
\begin{equation}\label{1a}
d_{G' +_x H'}(u,v)=d_{H'}(u,v)=d_{H}(u,v)=d_{G +_xH}(u,v),
\end{equation} 
for every pair $(u,v)\in V(H')\times V(H')$. An analogous argument shows that
\begin{equation}\label{1b}d_{G'+_xH'}(a,b)=d_{G+_xH}(a,b),
\end{equation} for every pair $(a,b)\in V(G')\times V(G')$. Next consider a pair $(a,u)\in V(G')\times V(H')$. Any geodesic path 
from $a$ to $u$ can be considered as the concatenation of an $a$--$x$ geodesic path in $G'$ and a $x$--$u$ geodesic path in $H'$. Applying Equations \eqref{1a} and \eqref{1b} implies that:
\begin{align*}
d_{G'+_x H'}(a,u)& =  d_{G' +_x H'}(a,x)+d_{G' +_x H'}(x,u)\\ &
= d_{G'}(a,x)+d_{H'}(x,u)\\ &
= d_{G}(a,x)+d_{H}(x,u) \\ &
= d_{G +_x H}(a,u).
\end{align*}
This completes the proof.
\end{proof}

To state the main result of this section, we use the following definition and notation.
\begin{defn}\label{def: DP}
For a graph $G$ and a vertex $x\in V(G)$, let
\begin{align*}
&\DP(G) = \{A \subseteq V(G): G[A] \le G\},\\
&\ddp(G) = \{|A| : A \in \DP(G)\},\\
&\ddp_x(G) = \{|A| : A \in \DP(G)\,\,\&\,\,x \not\in A\},\\
&\ddp^x(G) = \{|A| : A \in \DP(G)\,\,\&\,\,x \in A\}.\end{align*}
For sets $R$, $S$ and $T$ of integers, define $R+S+T := \{r+s+t : r \in R,\,\,s \in S\,\,\&\,\,t\in T\}$.
\end{defn}

\begin{thm}\label{dp(G+_xH)}
Consider a graph $G +_x H$. Then:
\begin{equation*}\label{eq:dp}
\ddp(G +_x H) \,\,\,=  \,\,\,\big(\ddp^x(G)+\ddp^x(H)+\{-1\}\big) \,\,\,\cup  \,\,\,\ddp_x(G) \,\,\,\cup\,\,\, \ddp_x(H). 
\end{equation*}
\end{thm} 
\begin{proof}
We consider two cases based upon whether $A\in \DP(G +_x H)$ contains~$x$. If $A$ does not contain $x$, then $A$ is fully contained in either $G$ or $H$, so $\ddp_x(G +_x H)=\ddp_x(G)\cup \ddp_x(H)$. If $A$ does contain $x$, then Lemma \ref{H'+G'} implies $A=G' +_x H'$, where $G'\le G$, $H'\le H$ and both contain $x$. Therefore, $\ddp^x(G +_x H)=\ddp^{x}(G)+\ddp^{x}(H)+\{-1\}$ where the minus~$1$ accounts for the common vertex $x$ in $G'$ and $H'$. Combining these two cases with the formula~$\ddp(G +_x H)=\ddp^x(G +_x H)\cup \ddp_x(G +_x H)$ completes the proof.
\end{proof}

\begin{figure}
  \begin{center}
    \begin{tikzpicture}	

    \shade[ball color=white] (4,0) circle (1.5cm);
    \shade[ball color=white] (-4,0) circle (1.5cm);
    \shade[ball color=black] (-4.2,.6) circle (.06cm);
    \shade[ball color=black] (3.5,.6) circle (.06cm);

    \node [right] at (-4.7,.6) {$x$};
    \node [right] at (3.5,.6) {$y$};
    \node [below] at (-4,.1) {$G$};
    \node [below] at (4,.1) {$H$};
    \node [below] at (-0.35,2.6) {$P(x,y)$};

    \draw (-4.2,.6) .. controls +(120:2cm) and +(80:2cm) .. (3.5,.6);
	\end{tikzpicture}
  \end{center}
 \caption[]{The figure for $G_x {\overset{r}{\text{---}}} H_y$.}\label{fig:figure2}
\end{figure}
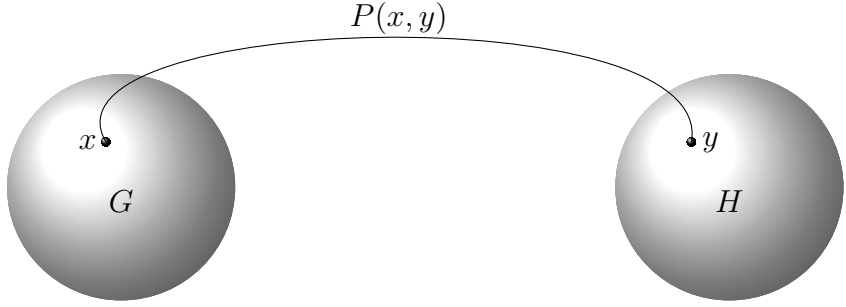

Given two disjoint graphs we can connect the graphs by a path of length~$r$, for any $r>0$.
\begin{defn}
Consider two disjoint graphs $G$ and $H$. Let $G_x {\overset{r}{\text{---}}} H_y$ be the graph obtained by connecting $x \in V(G)$ and $y \in V(H)$ with a path $P(x,y)$ of length $r$.
\end{defn}

%These graphs are less interesting than separable graphs.  And they are separable so that their dp-sets can be calculated by a simple iteration  of the previous result.
These graphs are separable, so applying a simple iteration of Theorem~\ref{dp(G+_xH)} gives the following corollary: 

\begin{cor}\label{dp(G+H)}
Consider two disjoint graphs $ G $ and $ H $. If $r>0$ then:

\begin{align*} 
\ddp(G_x {\overset{r}{\text{---}}} H_y)\ =\ & \big(\ddp^x(G)+\ddp^y(H) + \{-1,\dots ,r-1\}\big) \\
&\cup \ddp_x(G)\cup \ddp_y(H).
\end{align*}
\end{cor}

%%%%%%%%%%%%%%%%%%%%%%%%%%%%%%%%%%%%%%%%%%%%%%%%%%%%%%%%%%%%%%%%%%%%%%%%%%%%%%%%%%%%%%%%%%%%

\section{Maintaining the non-dp property}\label{$C_{k,l}$}
In this section we investigate the class of non-dp graphs. It is conjectured in~\cite{NF2} that almost all graphs are dp. So understanding this class is a logical step towards a full classification of the class of dp graphs. The simplest non-dp graphs are the cycle graphs~$C_k$, for all $k\ge 5$. We investigate how we can add vertices to the cycle graphs and preserve the non-dp property. To this end we introduce the following class of graphs.

Consider the cycle $ C_{k} $ and a set of vertices $A$, with $|A|=\ell$, such that $ A\cap V(C_{k})=\emptyset $. For each $a\in A$, select three consecutive vertices of $ C_{k} $ and  join~$a$ to at least one of the three selected vertices. Let $\mathcal{C}_{k,\ell}$ denote the family of graphs that can be constructed in this way. Given a graph $G\in \mathcal{C}_{k,\ell}$, let~$C(G)$ be the original cycle vertices of $G$ and $A(G)$ the added vertices. Note that the addition of the vertices to the cycle graph cannot change the distance between any pair of vertices in~$C(G)$, so~$C_k\le~G$.

Recall that we label the vertices of $C_k$ as $v_1,\ldots,v_k$, and let $v_{k+1}:=v_1$ and $v_0:=v_k$. So there is an edge between two vertices $v_i$ and $v_j$ if and only if $i=j\pm1$.

\begin{thm}\label{Emi}
If $k>2(\ell+2)$, then any graph in $\mathcal{C}_{k,\ell}$ is non-dp. 
\begin{proof}
Consider a graph $G \in \mathcal{C}_{k,\ell}$. If an added vertex $a$ is connected to two cycle vertices~$c_{i-1}$ and $c_{i+1}$, then the removal of either $a$ or $c_i$ results in isomorphic subgraphs. Therefore, when constructing an isometric subgraph of~$G$, by removing a set of vertices of $G$,  we can assume that~$a$ is always removed before $c_i$. Also recall that the added vertices do not alter the distance between any of the cycle vertices. Combining these two points implies that given a graph~$H\le G$ there is a geodesic path in $H$ between any two elements of $C(H)$ that is entirely contained in $H[C(H)]$. Therefore, if $H\le C_{k,\ell}$, then~$H[C(H)]\le C_k$. 

We show that there is no isometric subgraph of $G$ with order $\lfloor\frac{k}{2}\rfloor+2$. Suppose for a contradiction that such a subgraph does exist, we denote it~$H$. We know that $\ell<\frac{k}{2}-2$, so to obtain $H$ we must remove a set of $s$ cycle vertices, where $\lceil\frac{k}{2}\rceil-2>s>0$. However, this implies that $C(H)$ has~$t$ vertices, where $k>t>\lfloor\frac{k}{2}\rfloor+2$, and it is straightforward to see that there is no isometric subgraph of $C_k$ with $t$ vertices. Therefore, $H$ is not isometric, so~$G$ is non-dp.
\end{proof}
\end{thm}
 
 \begin{figure}\begin{center}\begin{tikzpicture}
    	\draw[line width=.3pt] (.6,1.9) to  (-.6,1.9);
    	\draw[line width=.3pt] (.6,-1.9) to  (-.6,-1.9);
    	\draw[line width=.3pt] (.6,1.9) to  (2.8,2.8);
    	\draw[line width=.3pt] (2.8,2.8) to  (-.6,1.9);
    	\draw[line width=.3pt] (2.8,2.8) to  (1.6,1.2);
    	\draw[line width=.3pt] (.6,1.9) to  (1.6,1.2);
    	\draw[line width=.3pt] (2,.05) to  (1.6,1.2);
    	\draw[line width=.3pt] (2,.05) to  (1.6,-1.2);
    	\draw[line width=.3pt] (2,.05) to  (2.8,-2.8);
    	\draw[line width=.3pt] (2.8,-2.8) to  (.6,-1.9);
    	\draw[line width=.3pt] (2.8,-2.8) to  (1.6,-1.2);
    	\draw[line width=.3pt] (.6,-1.9) to  (1.6,-1.2);
    	\draw[line width=.3pt] (-2.8,-2.8) to  (-.6,-1.9);
    	\draw[line width=.3pt] (-2.8,-2.8) to  (-1.6,-1.2);
    	\draw[line width=.3pt] (1.6,-1.2) to  (.6,-1.9);
    	\draw[line width=.3pt] (-1.6,-1.2) to  (-.6,-1.9);
    	\draw[line width=.3pt] (-1.6,-1.2) to (-2,.05);
    	\draw[line width=.3pt] (-2.8,-2.8) to  (-2,.05);
    		\draw[line width=.3pt] (-1.6,1.2) to  (-2,.05);
    		\draw[line width=.3pt] (-1.6,1.2) to  (-.6,1.9);
    		
    		\node [below] at (0,.3) {$G$};

    	\shade[shading=ball, ball color=black] (.6,1.9) circle (.1);
    	\shade[shading=ball, ball color=black] (-.6,1.9) circle (.1);
    	\shade[shading=ball, ball color=black] (.6,-1.9) circle (.1);
    	\shade[shading=ball, ball color=black] (-.6,-1.9) circle (.1);

    	\shade[shading=ball, ball color=black] (2,.05) circle (.1);

    	\shade[shading=ball, ball color=black] (1.6,1.2) circle (.1);
    	\shade[shading=ball, ball color=black] (1.6,-1.2) circle (.1);

    	\shade[shading=ball, ball color=black] (-1.6,1.2) circle (.1);
    	\shade[shading=ball, ball color=black] (-1.6,-1.2) circle (.1);

    	\shade[shading=ball, ball color=black] (-2,.05) circle (.1);

    	\shade[shading=ball, ball color=black] (2.8,2.8) circle (.1);

    	\shade[shading=ball, ball color=black] (2.8,-2.8) circle (.1);
    	\shade[shading=ball, ball color=black] (-2.8,-2.8) circle (.1);

		\end{tikzpicture}\end{center}
    \caption{A counterexample to the converse of the Theorem~\ref{Emi}}\label{fig:figure1001}
\end{figure}
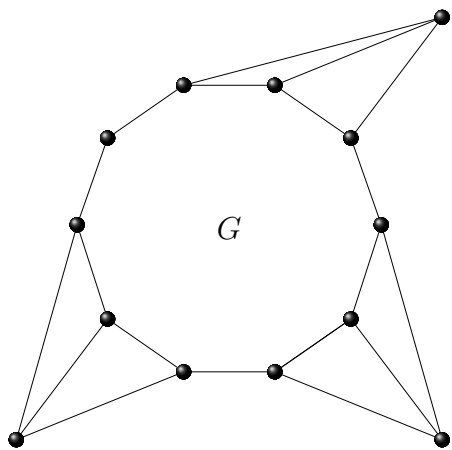

Note that the converse of Theorem~\ref{Emi} is not true. For example in Figure~\ref{fig:figure1001}, 
the graph~$G$ is not dp, since there is no isometric subgraph of order 9, but $k=10\not>10=2(\ell+2)$.

By Theorem~\ref{thm:sdp} we know that all non-dp graphs must contain a cycle of length $n\ge5$. We propose that the above approach can be extend to give a 
 constructive algorithm to build all non-dp graphs from a collection of cycles joined together by sequentially adding new vertices. 
 However, how to achieve this construction we leave as an open problem. 

\mbox{}
%An interesting question, which we leave open, is: %can we add vertices to non-cycle non-dp graphs in a similar way to get further results on the class of non-dp graphs.
%
%\begin{open}
%How can we add vertices to non-cycle non-dp graphs to get further results on the class of non-dp graphs?
%\end{open}

%We finish with an interesting conjecture about distance preserving graphs. Nussbaum and Esfahanian~\cite{NF2} have shown that if the minimum degree satisfies $\delta(G)\ge \frac{2n}{3}-1$, then $G$ is dp. 
% However, it is not clear whether or not this bound is tight.
% The largest known value of $\delta(G)$ in a non-dp graph is $\frac{n-1}{2}$. One example of such a graph is the $5$ cycle $C_5$, where $\delta(C_5)=2=\frac{n-1}{2}$ and $C_5$ is not dp.
% This leads us to the following conjecture, which has been verified for all graphs with $n\le10$ vertices using SageMath \cite{sagemath}.
%
%\begin{conj}\label{conj:degree}
%If  $G$ is an $n$-vertex graph with minimum degree  $\delta(G) \ge \frac{n}{2}$, then $G$ is dp. 
%\end{conj} 
%If $\delta(G) \ge \lfloor \frac{n}{2} \rfloor$ then $G$ has diameter at most~$2$. Since the possible distances in such a graph are so limited, one might be able to find the required isometric subgraphs.

%\bibliography{references}
%\bibliographystyle{abbrv}
\end{document}